\documentclass{huber_article}
\usepackage[utf8]{inputenc}

\usepackage{amsmath,amsfonts,amsthm}
\usepackage{dsfont}
\usepackage{graphicx}
\usepackage{color}
\usepackage[margin=1.5in]{geometry}
\usepackage{tikz}

\newcommand{\prob}{\mathbb{P}}
\newcommand{\mean}{\mathbb{E}}
\newcommand{\var}{\mathbb{V}}
\newcommand{\real}{\mathbb{R}}
\newcommand{\ind}{\mathds{1}}
\newcommand{\median}{\operatorname{med}}
\newcommand{\unifdist}{\textsf{Unif}}
\newcommand{\betadist}{\textsf{Beta}}
\newcommand{\normdist}{\textsf{N}}

\newtheorem{lemma}{Lemma}
\newtheorem{theorem}{Theorem}
\theoremstyle{definition}
\newtheorem{definition}{Definition}

\begin{document}

\title{An optimal 
$(\epsilon,\delta)$-approximation scheme
for the mean of random variables with
bounded relative variance}

\author{Mark Huber}

\date{}

\maketitle

\begin{abstract}
Randomized approximation algorithms for many 
\#P-complete problems (such as the partition
function of a Gibbs distribution, the volume of 
a convex body, the permanent of a $\{0,1\}$-matrix, and many
others)
reduce to creating random variables 
$X_1,X_2,\ldots$ with finite mean $\mu$ and standard deviation
$\sigma$ such that $\mu$ is 
the solution for the problem input, and the relative standard deviation $|\sigma/\mu| \leq c$ for known $c$.  Under these circumstances, 
it is known that the
number of samples from the $\{X_i\}$ needed to 
form an $(\epsilon,\delta)$-approximation $\hat \mu$ that satisfies
$\prob(|\hat \mu - \mu| > \epsilon \mu) \leq \delta$
is at least $(2-o(1))\epsilon^{-2} c^2\ln(1/\delta)$. 
We present here an easy to implement $(\epsilon,\delta)$-approximation 
$\hat \mu$ that uses $(2+o(1))c^2\epsilon^{-2}\ln(1/\delta)$ samples.
This achieves the same optimal running time as other estimators, but 
without the need for extra conditions such as bounds on third or
fourth moments.
\end{abstract}

\section{Introduction}

Suppose $X_1,X_2,\ldots \sim X$ are iid with mean $\mean[X] = \mu$ and variance
$\var(X) = \sigma^2$.  The relative standard deviation is $\sigma/|\mu|$ and the relative variance is $\sigma^2/\mu^2$.  Say the relative standard deviation
is bounded by $c$ if 
\begin{equation}
\label{EQN:bounded}
\left|\frac{\sigma}{\mu}\right| \leq c.
\end{equation}

Suppose $\mu$ and $\sigma$ are unknown, but $c$ is known.  Then the goal 
is to use as few $X_i$ as possible to find an estimate $\hat \mu$
for $\mu$ that is an $(\epsilon,\delta)$-randomized approximation, that is
\[
\prob(|\hat \mu - \mu| > \epsilon \mu) \leq \delta.
\]

Suppose that an $(\epsilon,\delta)$-randomized approximation requires
\[
(S + o(1)) c^2 \epsilon^{-2} \ln(1/\delta)
\]
samples for a constant $S$ and where the little-o notation refers to 
$\epsilon \rightarrow 0$.  
Then call $S$ the {\em scale factor} of the algorithm.

In this work we present a simple algorithm that is both easy to implement and
which achieves the optimal scale factor $S = 2$ 
without any additional assumptions about the random variables 
such as bounded higher moments.

This basic problem arises often in randomized 
algorithms.  For instance, problems for 
approximating the partition function of the Ising
model~\cite{jerrums1993}, the permanent of a 
$\{0,1\}$-matrix~\cite{jerrumsv2004}, the volume
of a convex body~\cite{dyerfk1991,kannanls1997}, the number of 
solutions to a DNF logical expression~\cite{jerrumvv1986}, the number of
linear extensions of a poset~\cite{matthews1991,huber2006b}, and many
more all have this problem as a subproblem.  Any
improvement in the ability to deal with this basic
problem
directly translates into better approximation
algorithms for all of these problems.

This problem has a long history, stretching back to 
Nemirovsky and Yudin~\cite{nemirovskyy1983} who used the median-of-means
estimator in the context of stochastic optimization.  Jerrum, Valiant,
and Vazirani~\cite{jerrumvv1986} developed a similar estimator for
the purposes of creating randomized approximation schemes for \#P complete
problems.  By 1999~\cite{alonms1999}, this method was in wide use for
online algorithms.  Hsu and Sabato~\cite{hsus2014,hsus2016} 
analyzed the 
basic median-of-means estimator and proved that it had a scale factor of  
121.5 for small enough $\epsilon$.

Catoni~\cite{catoni2012} greatly advanced the area by presenting an
approximation that used an $M$-estimator.  This was not an 
$(\epsilon,\delta)$-randomized approximation algorithm, rather it gave
a confidence bound based on the samples and specific values of the
parameters used in the estimate.  While it
could bound the confidence interval for the estimate
based on the parameters and the unknown $\mu$ and 
$\sigma^2$ for $X_i$, there seems to be no way of setting
the parameters ahead of time for given $\epsilon$ and $\delta$
without some additional
information.  Certainly no such method was given in~\cite{catoni2012}.

Devroye et. al.~\cite{devroyello2016} showed that if the kurtosis of
the random variables is bounded above, then the optimal scale factor
$S = 2$ could be attained with a simpler estimator.  
Unfortunately, in order to run their algorithm,
the user needed this upper bound on the kurtosis.  Bounding 
the kurtosis can be much more challenging mathematically than bounding
the variance.
Minsker and Strawn~\cite{minskers2017} 
returned to the original median-of-means estimator.  When the 
random variables have bounded third moment, the Berry-Esseen Theorem
can be used to show quick convergence to normality, and they showed
that this gave the simple median-of-means algorithm
a scale factor of 4.5.  As with
the Devroye et. al. method, this requires that the bound on the 
third moment be given explicitly before the algorithm can be used.

The approach here takes the Catoni $M$-estimator in a new direction.
There is no unique approach to getting the extra information
required to turn the Catoni 
$M$-estimator into an $(\epsilon,\delta)$-approximation.  One
approach is to use a two-step process that works as follows.  
Before running the $M$-estimator, first generate
a weaker estimate $\hat \mu_1$ that is an
$(\sqrt{\epsilon},\delta/2)$ randomized approximation to $\mu$.  
Then use this estimate
to set the parameters of the Catoni $M$-estimator
to give an output that is provably an
$(\epsilon,\delta)$-approximation.

While this two-step process works, it (like all $M$-estimators)
requires finding the root of a nonlinear equation.  Analysis of the
number of steps needed to get an close approximation to the root
was not done in~\cite{catoni2012}, and would need to be accomplished before the
running time of the method is known.

Previous algorithms either had too large
a scale factor, required rootfinding, or required knowledge of 
higher moments. The new method presented here solves all these difficulties.
\begin{itemize}
  \item It achieves the optimal scale factor $S = 2$.
  \item No rootfinding step is required.  Instead, first a function is
        randomly chosen by some initial samples, and then the final estimator
        is a sample average this random function applied to new data.
  \item No bound on higher moments is necessary.  In fact, even if the
        second moment is the highest moment that exists for the random
        variables, the new method is still an $(\epsilon,\delta)$-randomized
        approximation.
\end{itemize}

Our main result concerning this new method is as follows.
\begin{theorem}
  \label{THM:main}
  Let $X_1,X_2,\ldots \sim X$ with $\mean[X] = \mu$ and $\var[X] = \sigma^2$
  satisfying $\sigma^2/\mu^2 \leq c^2.$  For $\epsilon < 1$, there exists
  an $(\epsilon,\delta)$-randomized approximation algorithm that uses $n$
  samples where 
  \[
    n = \left\lceil\frac{2 c^2\epsilon^{-2}\ln(4/\delta)}
    {1 - \epsilon}\right\rceil + 
    \left\lceil 8\epsilon^{-1} \left(1+c^2\right) \right\rceil \cdot \left[2\left\lceil  \ln\left(\frac{7}{48\sqrt{\pi}\delta}\right)/\ln\left(\frac{16}{7}\right)\right\rceil 
    + 1\right].
  \]
\end{theorem}

This constant in the leading order term is the best possible.

\begin{theorem}
\label{THM:main2}
  Given $\epsilon$ and $\delta$ positive,
  let $\hat \mu:\real^n \rightarrow \real$ be an
  $(\epsilon,\delta)$-randomized
  algorithm for 
  all distributions $X$ with $\mean[X] = \mu$ and $\var(X) = \sigma^2$
  satisfying $\sigma^2/\mu^2 = c^2$.  Then 
  \[
     n \geq 2 \epsilon^{-2}c^2\left[\ln\left(\frac{1}{\sqrt{2\pi}\delta}\right) - 
       \ln\left(\frac{2\ln(1/[\sqrt{2\pi}\delta])+1}{\sqrt{2\ln(1/[\sqrt{2\pi}\delta])}}
  \right) \right].
  \]
\end{theorem}

The remainder of the paper is organized as follows.  The next section describes
the two-step algorithm and proves correctness for each step.  Section~\ref{SEC:lower}
then shows the lower bound on the number of samples needed, and 
Section~\ref{SEC:applications} considers several of the applications mentioned in
the introduction in more detail.

\section{The Algorithm}

Define the $\Psi$ function as follows.
\[
\Psi(u) = \ln(1 + u + u^2/2)\ind(u \geq 0)
         -\ln(1 - u + u^2/2)\ind(u \leq 0).
\]
This function was used in~\cite{catoni2012} as part of the $M$-estimator.

For $u \in [-1,1]$, the value of $\Psi(u)$ is approximately $u$ 
(see Figure~\ref{FIG:psi}.)  For
$u$ greater than 1 in magnitude, the value of $\Psi(u)/u$ becomes close to 0.  
For a constant $\alpha > 0$, $\alpha^{-1} \Psi(\alpha u)$ is a scaled version
of $\Psi$ that is close to 
$u$ for $u \in [-1/\alpha,1/\alpha]$.

Suppose that $\hat \mu_1$ is an initial estimate for $\mu$.  Then
\[
X_i = \hat \mu_1 + (X_i - \hat \mu_i)
\]
has mean $\mu$ but is also susceptible to outliers in the $X_i$
distribution.  By replacing this with
\[
W_i = \hat \mu_1 + \alpha^{-1}\Psi(\alpha\cdot(X_i - \hat \mu_i)),
\]
the value of $W_i$ will be close to $X_i$ when $|X_i - \hat \mu_i| \leq \alpha^{-1}$, but always has a light-tailed distribution because of the 
logarithm function.

\begin{center}
\begin{figure}[ht!]
  \begin{center}
  \begin{tikzpicture}[yscale=1.1,xscale=1.6]
    \draw (-2.1,0) -- (2.1,0);
    \draw (0,-2.1) -- (0,2.1);
    \draw (-2,0.1) -- (-2,-0.1) node[below] {-2};
    \draw (2,0.1) -- (2,-0.1) node[below] {2};
    \draw (0.1,2) -- (-0.1,2) node[left] {2};
    \draw (0.1,-2) -- (-0.1,-2) node[left] {-2};
    \draw[dashed,color=red] plot[domain=-2.1:2.1] (\x,\x);
    \draw[color=blue] plot[domain=0:2.1] (\x,{ln(1+\x+\x*\x/2});
    \draw[color=blue] plot[domain=-2.1:0] (\x,{-ln(1-\x+\x*\x/2});
    \draw (2.2,1.65) node[anchor=west] {$y=\Psi(x)$};
    \draw (2.2,2.1) node[anchor=west] {$y=x$};
  \end{tikzpicture}
  \end{center}
  \caption{The functions $\Psi(x)$ and $x$ over $[-2,2]$.}
  \label{FIG:psi}
\end{figure}
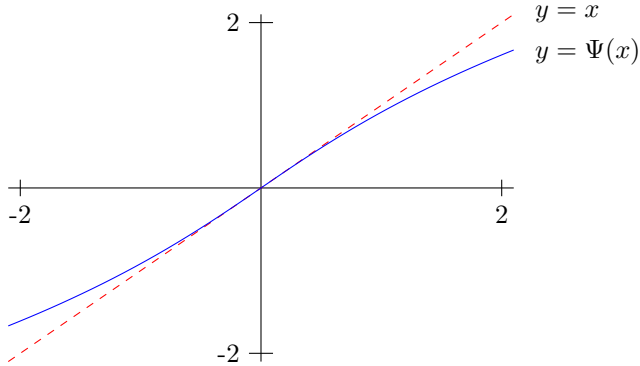
\end{center}

The algorithm proceeds as follows.  The first step uses a median-of-means
approach to find $\hat \mu_1$ that is a $(\sqrt{\epsilon},\delta/2)$
approximation of $\mu$.  Given that the first step
did not fail, the next step then uses the sample average
of the $W_i$ variables to create the final estimate $\hat \mu$ that is an $(\epsilon,\delta/2)$
approximation.  The chance that either step fails
is at most $\delta$.
\begin{enumerate}
  \item  
  The first step is to construct a median-of-means estimator for 
  $\mu$~\cite{jerrumvv1986}.
  Let $\epsilon_1 = \sqrt{\epsilon (c^2/(1+c^2))},$ 
  $k = \lceil 8c^2\epsilon_1^{-2}\rceil,$ and 
  $m = 2\lceil  \ln(7/[48\sqrt{\pi}\delta)/\ln(16/7)\rceil + 1$.  
  Let $S$ have the distribution of
  the sample average of $k$ independent draws from $X$.  
  Draw $S_1,\ldots,S_m \sim S$
  independently,
  and let $\hat \mu_1 = \text{median}(\{S_i\})/(1-\epsilon_1^2)$.  
  \item
  Let $n = \lceil 2c^2\epsilon^{-2}\ln(4/\delta)/(1-\epsilon)\rceil$, and
  $\alpha = \epsilon/[c^2 \hat \mu_1]$.  Draw
  $X_{1},\ldots,X_{n}$
  independently.  For all $i$, let
  \[
  W_i = \hat \mu_1 + \alpha^{-1}\Psi(\alpha \cdot 
    (X_i - \hat \mu_1)).
  \]
  Set $\hat \mu = (W_1 + \cdots + W_n)/n$.
\end{enumerate}

\subsection{The first step of the algorithm}
The first step of the algorithm is
the powering method of 
Jerrum, Valiant, and Vazirani~\cite{jerrumvv1986}
applied to the sample averages.  This technique
was also used in~\cite{alonms1999},
and later referred to as the median-of-means method~\cite{hsus2014}.  These
authors did not attempt to optimize
the constants in their arguments, and 
so we repeat the proof here so we can
see exactly how the choice of constant
enters into the failure bound.

Suppose that we have random variables whose 
relative standard deviation is at most $\nu \epsilon$.  What is
the chance that the median of $m = 2r + 1$ draws from the random
variable falls into $[\mu(1-\epsilon),\mu(1+\epsilon)]$?

To answer this question, first consider the probability that 
a beta distributed random variable with both parameters equal to an integer
$r$ falls into a subinterval of $[0,1]$.

\begin{lemma}
  Let $B \sim \betadist(r+1,r+1)$ denote a random variable with density 
  $f_B(x) = [(2r+1)!/(r!r!)]x^r(1-x)^r\ind(x \in [0,1])$.  For any
  $0 \leq a \leq 1/2 \leq b \leq 1$ with $1 - (b-a) \leq 1/2$,
  \[
  \prob(B \notin [a,b]) \leq 2\frac{4^r}{\sqrt{\pi r}} \cdot  
    \frac{[(1-(b-a))(b-a)]^{r+1}}{2(b - a)-1}
  \]
\end{lemma}

\begin{proof}
  Density $f_B$ is symmetric about its 
  unique local maximum at $1/2$, so
  \[
  \int_{x \in [0,a] \cup [b,1]} f_B(r)\ dr \leq 
    \int_{x \in [0,a] \cup [a,a+1-b]} f_B(r) =  \int_{x \in [0,1-(b-a)]} f_B(r) \ dr.
  \]
  
  Note $x^r(1-x)^r = (x - x^2)^r$.  Let $t = 1 - (b - a).$  Then 
  $t \leq 1/2$,
  $[x - x^2]' = 1 - 2x > 0$ and $[x - x^2]'' = -2x \leq 0$, 
  so the function lies below its tangent line at $t$.  That is, 
  \[
  (\forall x \in [0,t])(x - x^2) \leq t(1-t)+(x-t)(1-2t).
  \]
  
  Then 
  \[
  \int_{x \in [0,t]} [t(1-t)+(x-t)(1-2t)]^{r+1} 
     \ dr \leq \frac{[t(1-t)]^r}{(r+1)(1-2t)}. 
  \]
  Using Stirling's formula to give
  $(2r+1)(2r)!/(r!r!) \leq (2r+1)4^r/\sqrt{\pi r}$
  completes the proof.
\end{proof}

\begin{lemma}
  Let $A_1,\ldots,A_{2r+1}$ be iid with mean $\mu$ and variance
  at most $\nu^2 \epsilon^2 \mu^2$ where $\nu^2 \leq 1/2$.  Then
  \[
  \prob(\median(\{A_i\}) \notin [\mu(1-\epsilon),\mu(1+\epsilon)]) <
    \frac{(\nu^2)(1-\nu^2)}
        {\sqrt{\pi r}(1-2\nu^2)}
        \exp(r\ln(4(\nu^2)(1-\nu^2))).
  \]
\end{lemma}

\begin{proof}
  By Chebyshev's inequality, 
  \[
  \prob(A_i \notin [\mu(1-\epsilon),\mu(1+\epsilon)]) \leq \frac{\nu^2 \epsilon^2 \mu^2}{(\epsilon \mu)^2} = \nu^2.
  \]
  
  Let $x_1 = \prob(A_i < a)$ and $x_2 = \prob(A_i \leq b)$.  Construct a uniform random
  variable over $[0,1]$ as follows.  If $A_i < a$, then let $U_i \sim \unifdist([0,x_1))$.
  If $A_i \in [a,b]$, let $U_i \sim \unifdist([x_1,x_2])$.  Finally, if $A_i > b$, then
  let $U_i \sim \unifdist((x_2,1])$.  Note that  
  \[
  \prob(\median(\{U_i\}) \in [x_1,x_2]) = \prob(\median(\{A_i\}) \in [a,b]).
  \]
  The median of $2r+1$ iid uniform $[0,1]$ random variables is   well known to have a beta distribution: 
  $\median(\{U_i\}) \sim \betadist(r+1,r+1)$.  So the previous lemma can be used
  to state
  \[
  \prob(\median(\{U_i\}) \notin [x_1,x_2])
     \leq \frac{(\nu^2)(1-\nu^2)[4(\nu^2)(1-\nu^2)]^{r}}
        {\sqrt{\pi r}(1-2\nu^2)}
  \]
\end{proof}

Hence the failure probability is going down exponentially at rate
$\ln(4(\nu^2)(1-\nu^2)).$  Now for an integer $k$, consider
$X_1,\ldots,X_k \sim X$ iid, and 
\[
S = (X_1 + \cdots + X_k)/k.
\]
Then $\mean[S] = \mean[X] = \mu$ and 
$\var[S] = k \var[X]/k^2 = \sigma^2/k.$

In particular, for $\sigma^2/\mu^2 \leq c^2$, and $k = \lceil \epsilon^{-2}\nu^{-2}c^2\rceil$, 
$\var[S] \leq (\nu \epsilon \mu)^2.$  To take the median of 
$2r+1$ draws of the sample average of $k$ draws from the $\{X_i\}$ takes
$\Theta(kr) = \Theta(-1/(\nu^2 \ln(4\nu^2(1-\nu^2))))$ samples.

\begin{lemma}[Median-of-means]
  Suppose $X_1,X_2,\ldots$ are as in~\eqref{EQN:bounded}.
  For $k = \lceil 8c^2\epsilon^{-2}\rceil$ let 
  $S$ be distributed as $(X_1 + \cdots + X_k)/k$.  Let
  $m = 2\lceil \ln(7/[48\sqrt{\pi}\delta])/\ln(16/7)\rceil + 1.$
  Let $S_1,\ldots,S_m \sim S$ then 
  \[
  \prob(|\median(\{S_i\}) - \mu| > \epsilon \mu) \leq \delta.
  \]
\end{lemma}

\begin{proof}
   Just set $\nu^2 = 1/8$ in the previous lemma to get a rough 
   minimum for the bound.
\end{proof}

Now suppose that instead of bounding $\hat \mu_1/\mu  -1$, we wish
to bound $\xi = \mu/\hat \mu_1 - 1$.  The next lemma shows how to build
a biased estimate where $|\xi| \leq \epsilon$ from an estimate
with relative error at most $\epsilon$.
\begin{lemma}
  Suppose $\hat \mu_1$ is an estimate for $\mu$ with 
  $|\hat \mu_1/\mu - 1| \leq \epsilon$.  Then
  $\hat \mu_2 = \hat \mu_1/(1 - \epsilon^2)$ has 
  $|\mu/\hat \mu_2 - 1| \leq \epsilon$.
\end{lemma}

\begin{proof}
The proof follows from simplifying the appropriate inequalities.
\end{proof}

This is why in the first step of the algorithm, the estimate found
from median-of-means is 
divided by $1-\epsilon_1^2$ before moving to the next step.

\subsection{The second step of the algorithm}

To analyze this step, it helps to have
two new functions that upper and lower bound
$\Psi$.
\begin{align}
\Psi_U(x) = \ln(1 + x + x^2/2), \ 
\Psi_L(x) = -\ln(1-x+x^2/2).
\end{align}

\begin{lemma}
  For all $x \in \real$, 
  \[
  \Psi_L(x) \leq \Psi(x) \leq \Psi_U(x).
  \]
\end{lemma}

\begin{proof}
  First consider $\Psi_L(x) \leq \Psi(x)$.  These are equal when $x \leq 0$, so suppose $x \geq 0$.  Exponentiating gives
  \begin{align*}
    \Psi_L(x) \leq \Psi(x) &\Leftrightarrow [1-x+x^2/2]^{-1} \leq 1 + x + x^2/2. \\
     &\Leftrightarrow 1 \leq 1 + x^4/4,
  \end{align*}
  therefore the inequality holds.  The other inequality is shown similarly.
\end{proof}

Now set
\begin{align*}
  W_{L,i} &= \hat \mu_1 + \alpha^{-1}\Psi_L(\alpha \cdot 
    (X_i - \hat \mu_1)), \\
  W_{U,i} &= \hat \mu_1 + \alpha^{-1}\Psi_U(\alpha \cdot 
    (X_i - \hat \mu_1)).
\end{align*}
By the previous lemma, $W_{L,i} \leq W_i \leq W_{U,i}$ for all $i$.


\begin{lemma}
  Denote 
  $\bar W_U = (W_{U,1} + \cdots + W_{U,n})/n$.  Then
  \[
  \prob(\bar W > \mu(1+\epsilon)) \leq 
    \exp[-(n\epsilon^2/(2c^2))\cdot(1 - \xi^2(1 + 1/c^2)))].
  \]
\end{lemma}

\begin{proof}
  Take a Chernoff bound~\cite{chernoff1952} style approach.  Since $\alpha > 0$ and 
  $\exp$ is a strictly increasing function,
  \begin{align*}
  \prob(\bar W_U > \mu(1+\epsilon) &= \prob(W_{U,1} + \cdots + W_{U,n} > n \mu(1+\epsilon)) \\     
  &=  \prob(\exp(\alpha(W_{U,1} + \cdots + W_{U,n})) > \exp( \alpha n \mu(1 + \epsilon))) \\
  &\leq \mean[\exp(\alpha(W_{U,1} + \cdots + W_{U,n}))] / \exp(\alpha n \mu(1 + \epsilon)) \\
  &= \left[\frac{\mean[\exp(\alpha W_{U,1})]}{\exp(\alpha \mu(1 + \epsilon)}\right]^n
  \end{align*}
  First consider the expression inside the mean
  in the numerator.  Setting $\gamma = \mu - \hat\mu_1$ gives
  \begin{align*}
  \exp(\alpha W_{U,1}) &= \exp(\alpha \hat \mu_1 + \ln(1 + \alpha(X - \hat \mu_1)
     + (\alpha^2/2)(X - \hat \mu_1)^2) \\
     &= \exp(\alpha \hat \mu_1)[1 + \alpha(X - \hat \mu_1) + (\alpha^2/2)(X - \hat \mu_1)^2] \\
     &= \exp(\alpha \hat \mu_1)[1 + \alpha(X - \mu + \gamma) + (\alpha^2/2)(X - \mu + \gamma)^2] \\
     &= \exp(\alpha \hat \mu_1)[1 + \alpha(X - \mu) + \alpha \gamma + (\alpha^2/2)((X - \mu)^2 + 2(X-\mu)\gamma + \gamma^2)] 
  \end{align*}
  Since $\mean[X - \mu] = 0$ and $\mean[(X-\mu)^2] = \sigma^2$, we have
  \[
  \prob(\bar W_U > \mu(1 + \epsilon) \leq 
    \left[\frac{\exp(\alpha \hat \mu_1)[1 + \alpha \gamma + 
      \alpha^2\gamma^2/2 + \alpha^2 \sigma^2/2]}{\exp(\alpha\mu(1+\epsilon))}
      \right]^n
  \]
  
  Note 
  \[
  \frac{\exp(\alpha \hat \mu_1)}{\exp(\alpha\mu(1+\epsilon))}
    = \exp(-\alpha(\mu - \hat \mu_1) - \alpha \epsilon \mu)
    = \exp(-\alpha \gamma - \alpha \epsilon \mu).
    \]

Next use  $1 + x \leq \exp(x)$ to state
  \[
  \prob(\bar W_U > \mu(1 + \epsilon) \leq 
    \exp(-\alpha\gamma-\alpha\epsilon\mu+\alpha\gamma
      + \alpha^2\gamma^2/2 + \alpha^2\sigma^2/2)^n
  \]

Since $\alpha = \epsilon/[c^2\hat\mu_1^2]$, 
$\gamma = \mu - \hat \mu_1$ and $\xi = \mu/\hat\mu_1 - 1$,
\[
\frac{\alpha^2\gamma^2}{2} = \frac{(\mu - \hat\mu_1)^2}{2} \cdot
  \frac{\epsilon^2}{c^4 \hat\mu_1^2} = \left(\frac{\mu}{\hat\mu_1}-1\right)^2
  \cdot \frac{\epsilon^2}{2c^4} = \frac{\xi^2 \epsilon^2}{2 c^4}.
\]

Similarly, using $\sigma^2/\mu^2 \leq c^2$, 
\[
-\alpha\epsilon\mu + \frac{\alpha^2 \sigma^2}{2}
\leq - \alpha\epsilon\mu + \frac{\alpha^2 \mu^2 c^2}{2} 
 = - \frac{\epsilon^2}{c^2}\left[\frac{\mu}{\hat \mu_1} - \frac{1}{2}
 \left(\frac{\mu}{\hat \mu_1}\right)^2\right]
 = - \frac{\epsilon^2}{2c^2}\left[2(1+\xi)-(1+\xi)^2\right].
\]
Note $2(1+\xi)-(1+\xi)^2 = 1 - \xi^2$.

Putting this together with the $\alpha^2\gamma^2/2$ term gives
\begin{align*}
  \prob(\bar W_U > \mu(1 + \epsilon) &\leq 
    \exp\left(-\frac{\epsilon^2 n}{2c^2}
       \left(1 - \xi^2(1+1/c^2) \right)\right).
\end{align*}

\end{proof}

Note that at the end of Step 1 of the algorithm,
$\xi^2 \leq \epsilon(c^2/(1+c^2))$, 
which means $\xi^2(1 + 1/c^2) \leq \epsilon$, and 
$\prob(\bar W_U > \mu(1+\epsilon)) \leq 
  \exp(-n(1-\epsilon)\epsilon^2/[2c^2]).$

\begin{lemma}
  Denote 
  $\bar W_L = (W_{L,1} + \cdots + W_{L,n})/n$.  Then
  \[
  \prob(\bar W_L < \mu(1-\epsilon)) \leq 
    \exp(-(n\epsilon^2/(2c^2))\cdot
    (1 - \xi^2(1+1/c^2))).
  \]
\end{lemma}

\begin{proof}
    The proof is similar to the previous lemma:  first multiply
    by $-\alpha$ and exponentiate to get
  \begin{align*}
  \prob(\bar W_L > \mu(1-\epsilon) 
  &= \left[\mean[\exp(-\alpha W_{U,1})]\exp(\alpha \mu(1-\epsilon)\right]^n\\
  &= [\exp(\alpha\gamma-\alpha\epsilon\mu)(1-\alpha\gamma+\alpha^2\gamma^2/2+\alpha^2\sigma^2/2)]^n \\
  &\leq \exp(-\alpha\epsilon\mu+\alpha^2\gamma^2/2+\alpha^2\sigma^2/2)^n,
  \end{align*}
  and the rest of the proof is the same as the previous lemma.
\end{proof}

Putting these results together gives the following.
\begin{lemma}
  For $n \geq 2 c^2 \epsilon^{-2} \ln(2/\delta) 
        (1 - \epsilon)^{-1}$ and 
        $|\hat \mu_1 - \mu| \leq \sqrt{\epsilon c^2/(1+c^2)} \hat \mu_1$,
  \[
  \prob(|\bar W - \mu| > \epsilon \mu) \leq \delta.
  \]
\end{lemma}

\begin{proof}
  Apply the previous lemma using
  $\xi^2 \leq \epsilon c^2/(1+c^2).$
\end{proof}

Theorem~\ref{THM:main} immediately follows.

\section{Lower bound on the number of samples}
\label{SEC:lower}

Begin with a rephrasing of Proposition 6.1 from~\cite{catoni2012}.
\begin{lemma}
  Let $\hat \mu:\real^n \rightarrow \real$ be any estimator of the mean of
  $n$ iid random variables.  
  Let $Y_1,\ldots,Y_n \sim \normdist(\mu,\sigma^2)$ and $\bar Y = (Y_1 + \cdots + Y_n)/n$.  Then it holds that either
  $\prob(\hat \mu \geq \mu(1+\epsilon)) \geq \prob(\bar Y \geq \mu(1+\epsilon))$
  or
  $\prob(\hat \mu \leq \mu(1-\epsilon)) \geq \prob(\bar Y \leq \mu(1-\epsilon))$
  for $\bar Y = (Y_1 + \cdots Y_n)/n$.
\end{lemma}

In other words, for any estimator of the mean for normal random variables,
there is either a higher chance that the estimate is in the upper tail than
for the sample average, or there is a higher chance that the estimate
falls in the lower tail than the sample average does.
Note that for $Y_i \sim \normdist(\mu,c^2\mu^2)$, then 
$\bar Y \sim \normdist(\mu,c^2\mu^2/n)$.  Let $Z \sim \normdist(0,1)$.
From the scaling properties of normal random variables,
\[
\prob(\bar Y \in [\mu(1-\epsilon),\mu(1+\epsilon)]) = 
  \prob(Z \in [-\epsilon \sqrt{n}/c,\epsilon \sqrt{n}/c]).
\]
Since $\prob(Z \leq -a) = \prob(Z \geq a)$ for all $a$, we need only bound one
tail of the normal.

\begin{lemma}
  Let $Z \sim \normdist(0,1)$ and $a_{\delta}$ satisfy
  $\prob(Z \geq a_{\delta}) = \delta$ where $\delta \leq 1/\sqrt{2\pi}$.  Then
  \[
  a_{\delta}^2 \geq 2 \ln\left(\frac{1}{\sqrt{2\pi}\delta}\right) + 2\ln\left(\frac{\sqrt{2\ln(1/[\sqrt{2\pi}\delta])}}{2\ln(1/[\sqrt{2\pi}\delta])+1}
  \right).
  \]
\end{lemma}

\begin{proof}
  Gordon~\cite{gordon1941} showed that  for $a \geq 0$,
  \[
  \prob(Z \geq a) \geq \frac{a}{a^2+1}\frac{1}{\sqrt{2\pi}}\exp(-a^2/2).
  \]
  Without the $a/(a^2+1)$ factor, the right hand side equals $\delta/2$ when
  $a_1 = \sqrt{2\ln(1/(\sqrt{2 \pi}\delta))}$.  Since 
  $a/(a^2+1) \leq 1$ we have $a_{\delta} \leq a_1$.  
  Also, $a/(a^2+1)$ is a decreasing function, so
  \[
  \frac{a_1}{a_1^2+1}\frac{1}{\sqrt{2\pi}} \exp(-a_{\delta}^2/2) \leq \delta/2.
  \]
  Solving gives
  \[
  a_{\delta}^2 \geq 2\ln\left(\frac{a_1}{a_1^2+1}\frac{1}{\sqrt{2\pi}\delta} \right)
  \]
  as desired.
\end{proof}

Putting $a_{\delta} = \epsilon \sqrt{n}/c$ then gives Theorem~\ref{THM:main2}.

\section{Applications}
\label{SEC:applications}

Jerrum, Valiant, and Vazirani~\cite{jerrumvv1986} showed
that for a large class of {\em self-reducible} problems,
the ability to sample from a density
in polynomial time leads to an $(\epsilon,\delta)$-randomized 
approximation scheme for the normalizing constant of the 
unnormalized density.  Since finding that normalizing constant
is often a \#P-complete problem, this has been used in many
settings.  Each of these leads to a problem such as that 
considered here where a random variable has mean $\mu$ 
equal to the target with bounded relative standard deviation.  This
method was expanded to more examples later by Jerrum and 
Sinclair~\cite{jerrums1996}.

The idea is as follows.  Suppose that the goal is to 
find $\#A_0$ which is the size of a set (either number of elements for
a finite set or the Lebesgue measure for $A_0 \subset \real^n$.)
Suppose that we can find a sequence of decreasing sets 
$A_0 \supseteq A_1 \supseteq A_2 \supseteq \cdots \supseteq A_k$ where
$\#A_k$ is known.  If each of the sets $A_i$ represents an instance
of the original problem (perhaps with a different input), the problem
is self-reducible.  If there is an efficient method for generating samples
uniformly from the $A_i$, then for each $i \in \{0,1\ldots,k-1\}$, let
$X_{i,1},\ldots,X_{i,m} \sim \unifdist(A_i)$, and let 
$R_i = m^{-1}\sum_{j} \ind(X_{i,j} \in A_{i+1})$ be the percentage
of values that fall into $A_{i+1}$.  Then 
\[
\frac{\#A_k}{\#A_0} = \mean[R_0] \mean[R_1]\cdots \mean[R_{k-1}],
\]
so let $\hat r = R_0 \cdots R_{k-1}$ be the unbiased {\em product estimator}
for $\#A_k/\#A_0$.

Then 
\[
\frac{\var(\hat r)}{\mean[\hat r]^2} =
  \frac{\mean[\hat r^2]}{\mean[\hat r]^2} - 1 = 
  \left[\prod_{i=1}^k \frac{\mean[R_i^2]}{\mean[R_i]^2}\right] - 1 =
  \left[\prod_{i=1}^k \left(1 + \frac{\var[R_i]}{\mean[R_i]^2}\right)\right] - 1 
\]
Let $r_i = \#A_i/\#A_{i+1}$.  Then $X_{i,1}$ 
has a Bernoulli distribution with mean $r_i$ and variance $r_i(1-r_i)$.
As the sample average of $m$ iid draws from $X_{i,1}$, 
$R_i$ has mean $r_i$ and variance $r_i(1-r_i)/m$.  
Then
\[
\frac{\var(\hat r)}{\mean[\hat r]^2} \leq \left[\prod_{i=1}^k 1 + \frac{1-r_i}{m r_i}\right] - 1,
\]
so if $r_i \geq 1/M$ for all $i$, using $1 + x \leq \exp(x)$ gives
\[
\frac{\var(\hat r)}{\mean[\hat r]^2} \leq \exp\left(\frac{k (M-1)}{m}\right)-1.
\]

There are $k$ different $R_i$ each requiring $m$ samples, therefore $km$ are
needed to generate one value of $\hat r$.  From the above the 
variance is $\exp(k(M-1)/m)-1 \approx k(M-1)/m$ for large $m$.  Hence for 
large $m$ (such as $k (M-1) \epsilon^{-2}$) using the algorithm presented here
has the total number of samples needed
for an $(\epsilon,\delta)$-approximation is (to leading order)
$2 k(M-1) \epsilon^{-2} \ln(4/\delta)$, with the 2 being the optimal value of
the constant.

\subsection{Linear extensions of a poset}

For a direct application of this process, consider the problem of counting
the number of linear extensions of a partially ordered set (poset).  A
{\em poset} on $n$ objects $\{1,\ldots,n\}$ is an ordering $\preceq$ with three
properties.  Let $i,j,k \in \{1,\ldots,n\}$. First, $i \preceq i$.  Second, if $i \preceq j$
and $j \preceq i$, then $i = j$.  Third, if $i \preceq j$ and $j \preceq k$,
then $i \preceq k$.  A {\em linear extension} of the poset is a permutation
$\tau$ such that $\tau(i) \preceq \tau(j) \Rightarrow i \leq j$.

Brightwell and Winkler~\cite{brightwellw1991} showed that counting the number of
linear extensions of an arbitrary poset is a \#P-complete problem.  Finding
the number of linear extensions has applications in nonparametric statistics~\cite{mortonpssw2009}.  

A sequence of results~\cite{karzanovk1991,matthews1991,bubleyd1999,huber2006b} culminated
in an $O(n^3\ln(n))$ method for generating samples uniformly from the set of linear 
extensions.  To convert this method into a method for approximately counting
the number of linear extensions, use self-reducibility.

Let $n_\ell$ be any element of $\{1,\ldots,n\}$ which is not preceded by another
element in the set.  Then an easy Markov chain argument gives that the probability
that a uniformly chosen linear extension has $\tau(n_\ell) = n$ is at least $1/n$.
Fixing $\tau(n_\ell) = n$ in the permutation leaves a linear extension problem
of size $n - 1$.  So the methods of this section can be applied with $k = n$ and 
$M = n$.  Hence (to first order) $2n^2 \epsilon^{-2}\ln(4/\delta)$ samples are 
needed to give an $(\epsilon,\delta)$-approximation to the number of linear extensions.

\subsection{Permanent of a $\{0,1\}$-matrix}
Let ${\cal S}_n$ be the set of permutations on $\{1,\ldots,n\}$.  Then 
the permanent of a matrix $A$ with entries
$a_{ij}$ is 
\[
\sum_{\tau \in {\cal S}_n} \prod_{i=1}^n 
 a_{i,\tau(i)}.
\]
Calculating the permanent exactly was shown by Valiant~\cite{valiant1979}
to be a \#P-complete problem.

Note that if $a_{ij} \in \{0,1\}$, then the only permutations $\tau$ that
contribute to the sum have $a_{i,\tau(i)} = 1$ for all $i$.  So the permanent
is the normalizing constant of the distribution over ${\cal S}_n$ with
unnormalized density $f(\tau) = \prod_{i=1}^n a_{i,\tau(i)}.$

Jerrum, Sinclair, and Vigoda~\cite{jerrumsv2004} developed a polynomial
time algorithm for approximately sampling from the density $f(\tau)$.  
As with the previous problem of linear extensions, for such a problem
on permutations there exists a value $i$ such that $\prob(\tau(n) = i) \geq 1/n$.
This can then be used with the basic self-reducibility process to get
an $(\epsilon,\delta)$-approximation for the permanent.  Without going into
details (as the method of~\cite{jerrumsv2004} for approximation was more complex 
than the basic approach) the result is the same as for linear extensions:
use of the methods of 
this paper immediately
reduces the constant in the leading term down to the optimal value.

\subsection{Gibbs distributions}  These distributions arise
in statistical physics and other applications.

\begin{definition}
  $\{\pi_\beta\}_{\beta \in \real}$ is a {\em Gibbs distribution with parameter $\beta
  $} over finite state space $\Omega$ if there exists a {\em Hamiltonian function} $H(x):\Omega \rightarrow \real$ such that for $X \sim \pi_\beta$,
  \[
  \prob(X = x) = \exp(-\beta H(x))/Z(\beta),
  \]
  where $Z(\beta) = \sum_{x \in \Omega} \exp(-\beta H(x))$ is 
  called the {\em partition function} of the distribution.
\end{definition}

A famous example of a Gibbs 
distribution is the Ising model~\cite{ising1925}, 
where the state space consists of labellings of the nodes
of a graph $G = (V,E)$ by either 0 or 1, and 
$H(x) = \sum_{\{v,w\} \in E} -(x(v)-x(w))^2.$  In~\cite{jerrums1993} finding the partition function of
the ferromagnetic Ising model (where $\beta > 0$) was shown to be a \#P-complete problem for general graphs,
but that same work showed how to generate (approximately) samples from the distribution in time polynomial in the 
size of the graph.

Typically it is easy to find $Z(0)$ for these problems.  For the Ising model, $Z(0) = 2^{\#V}$.  In~\cite{stefankovicvv2009} it was shown how to build
an estimate for $Z_\beta/Z(0)$ using samples from $\pi$
where the ratio $\sigma^2/\mu^2$ was bounded.  
In~\cite{huber2015a}, it was shown how to build
two random variables $W$ and $V$ such that 
$\mean[W]/\mean[V] = Z(\beta)/Z(0)$ and
each had relative variance bounded above by
$2e$.

Let $\epsilon' = [-1 + \sqrt{1+\epsilon^2}]/\epsilon \leq \epsilon/2-\epsilon^3(1.5-\sqrt{2})$
for $\epsilon \in [0,1]$.  If 
\[
|\hat \mu_W -\mean[W]| \leq \epsilon \mean[W]
\text{ and }
|\hat \mu_V -\mean[V]| \leq \epsilon \mean[V],
\]
then
\[
\frac{\mean[W]}{\mean[V]}\frac{1 - \epsilon'}{1+\epsilon'} \leq \frac{\hat \mu_W}{\hat \mu_V}
  \leq \frac{\mean[W]}{\mean[V]}\frac{1+\epsilon'}
  {1-\epsilon'}.
\]

Then it is straightforward to show that 
$\hat \mu = [\hat \mu_W/\hat \mu_V]\sqrt{1+\epsilon^2}$ satisfies
\[
\hat \mu \in [(\mean[W]/\mean[V])(1-\epsilon),(\mean[W]/\mean[V])(1+\epsilon)],
\]
thereby
giving an $(\epsilon,\delta)$-approximation that (to leading order) requires
$2 (4 \epsilon^{-2})(2e)^2\ln(4/\delta)$ samples to estimate the partition function value.

\bibliographystyle{plain}

\end{document}